\documentclass[a4paper,UKenglish]{lipics-v2016}

\usepackage{wasysym}
\usepackage{amsthm}
\usepackage{amssymb}
\usepackage{bm}
\usepackage{slashbox}
\usepackage{graphicx}
\usepackage{enumerate}
\usepackage{multicol}
\usepackage{cleveref}
\usepackage[algo2e,ruled,vlined]{algorithm2e}
\usepackage{microtype}

\newtheorem{fact}{Fact}

\newcommand{\Prob}{\Pr}
\newcommand{\E}{\mathrm{E}}
\newcommand{\weakINV}{\textbf{($\clubsuit$)}}

\newcommand{\first}{\textbf{($\bullet$)}}
\newcommand{\second}{\textbf{($\circ$)}}
\newcommand{\INV}{\textbf{($\ast$)}}
\newcommand{\WindowSort}{\textsc{Window Sort}\xspace}

\newcommand{\err}[2]{ERRORS(#1,#2)}

\def\signed #1{{\leavevmode\unskip\nobreak\hfil\penalty50\hskip2em
		\hbox{}\nobreak\hfil(#1)%
		\parfillskip=0pt \finalhyphendemerits=0 \endgraf}}

\newsavebox\mybox
\newenvironment{aquote}[1]
{\savebox\mybox{#1}\begin{quote}}
	{\signed{\usebox\mybox}\end{quote}}

\title{Sorting with Recurrent  Comparison Errors\footnote{Research supported by SNSF (Swiss National Science Foundation, project number 200021\_165524).}}

\author{Barbara~Geissmann}
\author{Stefano~Leucci} 
\author{Chih-Hung~Liu}
\author{Paolo~Penna}
\affil{Department of Computer Science, ETH Z\"{u}rich, Switzerland
\texttt{\{name\}.\{surname\}@inf.ethz.ch}}

\authorrunning{B. Geissmann, S. Leucci, Ch. Liu and P. Penna}

\Copyright{Barbara Geissmann, Stefano Leucci, Chih-Hung Liu, and Paolo Penna}


\makeatletter
\let\@DOIPrefix\@empty
\makeatother

\begin{document}
\maketitle
\begin{abstract}
	We present a sorting algorithm for the case of \emph{recurrent} random comparison \emph{errors}. The algorithm essentially achieves simultaneously good properties of previous algorithms for sorting $n$ distinct elements in this model. In particular, it runs in $O(n^2)$ time,  the \emph{maximum dislocation} of the elements in the output is $O(\log n)$, while the \emph{total dislocation} is $O(n)$. These guarantees are the best possible since we prove that even randomized algorithms cannot achieve $o(\log n)$ maximum dislocation with high probability, or $o(n)$ total dislocation in expectation, regardless of their running time. 
\end{abstract}
\section{Introdution}
We study the problem of \emph{sorting} $n$ distinct elements under \emph{recurrent} random comparison \emph{errors}. In this classical model, 
each comparison is wrong with some fixed probability $p$, and correct with probability $1-p$. The probability of errors are independent over all possible pairs of elements, but errors are recurrent: If the same comparison is repeated at any time during the computation, the result is always the same, i.e., always wrong or always correct. 

In such a scenario not all sorting algorithms perform equally well in terms of the output, as some of them are more likely to produce ``nearly sorted'' sequences than others.
To measure the quality of an output permutation in terms of sortedness, a common way is to consider the \emph{dislocation} of an element, which is the difference between its position in the permutation and its true rank among all elements. Two criteria based on the dislocations of the elements are the \emph{total dislocation} of a permutation, i.e., the sum of the dislocations of all $n$ elements,
and the \emph{maximum dislocation} of any element in the permutation. 
Naturally,  the running \emph{time} remains an important criteria for evaluating sorting algorithms.  

To the best of our knowledge, for recurrent random comparison errors, the best results with respect to \emph{running time}, \emph{maximum}, and \emph{total dislocation} are achieved by the following two different algorithms:
\begin{itemize}
	\item Braverman and Mossel~\cite{Braverman2008} give an algorithm which guarantees maximum dislocation $O(\log n)$ and total dislocation $O(n)$, both with high probability. The main drawback of this algorithm seems to be its running time, as the constant exponent can be rather high;
	\item Klein et al.~\cite{Klein2011} give a much faster $O(n^2)$-time algorithm which guarantees maximum dislocation $O(\log n)$, with high probability. 	They however do not provide any upper bound on the total dislocation, which by the previous result is obviously $O(n\log n)$. 
\end{itemize}
\medskip
In this paper we investigate whether it is possible to guarantee  \emph{all} of these bounds together, that is, if there is an algorithm with running time $O(n^2)$,  maximum dislocation $O(\log n)$, and total dislocation $O(n)$.

\begin{table}[t]
\centering
\begin{tabular}{|l|c|c||c|}
\hline
& Braverman and Mossel \cite{Braverman2008} & Klein et al~\cite{Klein2011} & Ours\\
\hline
\# Steps & $O(n^{3+24c})$ & $O(n^2)$ & $O(n^2)$  \\
\hline 
Maximal Dislocation & w.h.p. $O(\log n)$ & w.h.p. $O(\log n)$ & w.h.p. $O(\log n)$ \\
\hline
Total Dislocation & w.h.p.  $O(n)$ & w.h.p. $O(n \log n)$  & in exp. $O(n)$ \\
\hline
\end{tabular}
\vspace{10pt}
\caption{Our and previous results. 
	The constant $c$ in \cite{Braverman2008} depends on both the error probability $p$, and the success probability of the algorithm. For example, for a success probability of $1-1/n$, the analysis in \cite{Braverman2008} yields   $c=\frac{110525}{(1/2-p)^4}$. The total dislocation in \cite{Klein2011} follows trivially by the maximum dislocation, and no further analysis is given.}
\label{tb-recurrent}
\end{table}

\subsection{Our contribution}
We propose a new algorithm whose performance guarantees are essentially the best of the two previous algorithms (see Table~\ref{tb-recurrent}). Indeed, our algorithm  \WindowSort
 takes $O(n^2)$ time and guarantees the maximum dislocation to be $O(\log n)$ with probability $1-1/n$
and the expected total dislocation to be $O(n)$. 
The main idea is to iteratively sort $n$ elements by comparing each element with its neighbouring elements lying within a \emph{window}
and to halve the window size after every iteration.
In each iteration, each element is assigned a rank based on the local comparisons,
and then they are placed according to the computed ranks.

Our algorithm is  inspired by Klein et al.'s algorithm \cite{Klein2011} which 
distributes elements into \emph{buckets} according to their computed ranks, compares each element with elements in neighboring buckets to obtain a new rank, and halves the range of a bucket iteratively. 
Note however that the two algorithms operate in a different way, essentially because of the following key difference between \emph{bucket} and \emph{window}. 
The number of elements in a bucket is not fixed, since the computed rank of several elements could be the same. In a window, instead, the number of elements is \emph{fixed}. This property is essential in the analysis of the total dislocation of  \WindowSort, but introduces a potential \emph{offset} between the computed rank and the computed position of an element. Our analysis consists in showing that such an offset is sufficiently small, which we do by considering a number of ``delicate'' conditions that the algorithm should maintain throughout its execution with sufficiently high probability. 

We first describe a standard version of our algorithm which achieves the afore mentioned bounds for any error probability $p <1/32$. We then improve this result to $p<1/16$ by using the idea of shrinking the window size at a different rate. Experimental results (see \Cref{sec-exp}) show that
the performance of the standard version is significantly better than the theoretical guarantees.
In particular, the experiments suggest that the expected total dislocation  is $O(n)$ for $p<1/5$,
while the maximum dislocation is $O(\log n)$ for $p<1/4$.

In addition, we prove that no sorting algorithm can guarantee the maximum dislocation to be $o(\log n)$ with high probability,
and no sorting algorithm can guarantee the expected total dislocation to be $o(n)$.

\subsection{Further Related Work on Sorting with Comparison Errors}
Computing with errors is often considered in the framework of a two-person game called \emph{R\'{e}nyi-Ulam Game}:
The \emph{responder} thinks of an object in the search space, and the \emph{questioner} has to find it by asking questions to which the responder provides answers.
However, some of the answers are incorrect on purpose; the responder is an \emph{adversarial lier}.
These games have been extensively studied in the past on various kinds of search spaces, questions, and errors; see Pelc's survey \cite{Pelc02} and Cicalese's monograph \cite{Cicalese13}.

Feige et al.~\cite{Feige1994} studied several comparison based algorithms with independent random errors
where the error probability of a comparison is less than half, the repetitions of an comparison can obtain different outcomes, and all the comparisons are independent.
They required the reported solution to be correct with a probability $1-q$, where $0<q<1/2$,
and proved that for sorting,  $O(n\log(n/q))$ comparisons suffice, which gives also the running time. 
In the same model, sorting by random swaps represented as Markovian processes have been studied under the question of the number of inversions (reversed pairs) \cite{gecco17,barbarapaolo}, which is within a constant factor of the total dislocation \cite{diaconis}.
Karp and Kleinberg~\cite{KarpK07} studied a noisy version of the classic binary search problem, where elements cannot be compared directly. 
Instead, each element is associated with a coin that has an unknown probability of observing heads when tossing it and these probabilities increase when going through the sorted order.

For recurring errors,
Coppersmith and Rurda \cite{Coppersmith} studied a simple algorithm that gives a 5-approximation on the weighted feedback arc set in tournaments (FAST) problem if the weights satisfy probability constraints. The algorithm consists of ordering the elements based on computed ranks, which for unweighted FAST are identical to our computed ranks.
Alonso et al.\ \cite{alonso} and Hadjicostas and Lakshamanan \cite{hadji} studied Quicksort and recursive Mergesort, respectively, with random comparison errors.

\paragraph*{Paper organization}
The rest of this paper is organized as follows. We present the \WindowSort\ algorithm in Section~\ref{sec-algorithm} and analyze the maximum and total dislocation in 
Section~\ref{sec-max} and Section~\ref{sec-total}, respectively. 
Then, we explain how to modify \WindowSort\ to allow larger error probabilities in \Cref{sec-extension} and show some experimental results in \Cref{sec-exp}. Additionally, we prove lower bounds on both the maximum and average dislocation for any sorting algorithm in \Cref{sec:lowerbound}.

\section{Window Sort}\label{sec-algorithm}
\WindowSort\ consists of multiple iterations of the same procedure: 
Starting with a permutation $\sigma$ and a \textit{window size} $w$,  
we compare each element $x$ in $\sigma$ with its left $2w$ and right $2w$ adjacent elements (if they exist) and count its \emph{wins}, i.e., the number of times a comparison outputs $x$ as the larger element. Then, we obtain the \emph{computed rank} for each element based on its \textit{original position} in $\sigma$ and its wins: 
if $\sigma(x)$ denotes the original position of $x$ in $\sigma$, the computed rank of $x$ equals $\max\{0,\sigma(x)-2w\}$ plus the number of its wins.
And we get a new permutation $\sigma'$ by placing the elements ordered by their computed ranks.
Finally, we set $w' = w/2$ and start a new iteration on $\sigma'$ with window size $w'$. In the very first iteration, $w=n/2$.
We formalize \WindowSort\ in Algorithm \ref{alg:iter}.

In the following, w.l.o.g. we assume to sort elements $\{1,\dots,n\}$, i.e., we refer to both an element $x$ and its rank by $x$.
Let $\sigma$ denote the permutation of the elements at the beginning of the current iteration of \WindowSort\ and let $\sigma'$ denote the permutation obtained after this iteration (i.e., the permutation on which the next iteration performs). Similarly, let $w$ and $w'=w/2$ denote the window size of the current and the next iteration.
Furthermore, let $\pi$ denote the sorted permutation. We define four important terms for an element $x$ in $\sigma$: 
\begin{itemize}
	\item \emph{Current/Original position}: The position of  $x$  in $\sigma$: $\sigma(x)$
	\item \emph{Computed rank}: The current position of $x$ minus $2w$ (zero if negative) plus its number of wins: $computed\_rank(x)$
	\item \emph{Computed position}: The position of  $x$  in $\sigma'$: $\sigma'(x)$
\end{itemize}

\begin{theorem}
	\WindowSort~takes $O(n^2)$ time. \end{theorem}
\begin{proof}
	Consider the three steps in Algorithm \ref{alg:iter}.
	The number of comparisons in an iteration in the outer loop is $4w$, for $w$ the current size of the window. Therefore, the first step needs $O(nw)$ time. For the second step we could apply for instance Counting Sort (see e.g. \cite{Cormen}), which takes $O(n)$ time, since all computed ranks lie between zero and $n$. Thus, the total running time is upper bounded by
	$\sum_{i=0}^\infty O(\frac{4n^2}{2^i}) = O(8 n^2)$. 	
\end{proof}

\begin{algorithm2e}[t]
	\vspace{2mm}
	\textbf{Initialization:} The initial window size is $w=n/2$. Each element $x$ has two variables $wins(x)$ and $computed\_rank(x)$ which are set to zero.
	
	\Repeat{$w<1$}{
		\begin{enumerate} 
			\item \ForEach{\emph{{$x$ at position}} $l=1,2,3,\ldots,n$\emph{ in }$\sigma$}{
				\ForEach{\emph{{$y$ whose position in} $\sigma$  is in} $[l-2w,l-1]$ or \emph{in} $[l+1,l+2w]$}{
					\If{$x>y$}{
						$wins(x) = wins(x) + 1$}}
				$computed\_rank(x) = \max\{l-2w, 0\} + wins(x)$
			} \label{alg:computed_rank}
			\item {Place} the elements into $\sigma'$ ordered by non-decreasing $computed\_rank$,\\
			break ties arbitrarily.
			\item Set all $wins$ to zero, $\sigma=\sigma'$, and $w=w/2$.
		\end{enumerate}	}
		\caption{\WindowSort~ (on a permutation $\sigma$ on $n$ elements)}
		\label{alg:iter}
	\end{algorithm2e}

\subsection{Preliminaries}

We first introduce a condition on the errors in comparisons between an element $x$ and a fixed subset of elements which depends only on the window size $w$. 

\begin{definition}
	We define $\err{x}{w}$ as the set of errors among the comparisons between $x$ and every $y\in [x-4w,x+4w]$. 
\end{definition}

\begin{theorem}\label{th:window-sort:error-analysis}
	 \WindowSort returns a sequence of maximum dislocation at most $9w^{\star}$ whenever the initial comparisons are such that 
	\begin{equation}
	|\err{x}{w}| \leq w/4 \label{eq:error-condition}
	\end{equation}
	hold for all elements $x$ and for all $w=n/2,n/4,\ldots,2w^\star$.
\end{theorem}
The proof of this theorem follows in the end of this section.
In the analysis, we shall prove the following:
\begin{itemize}
	\item If the computed rank of each element is close to its (true) rank, then the dislocation of each elements is small (Lemma~\ref{lem-diff}). 
	\item The computed rank of each element is indeed close to its (true) rank if the number of errors involving the element under consideration is small (Lemma~\ref{lem-analyzable}).
	\item The number of positions an element can move in further iterations is small (Lemma~\ref{lem-offset}).
\end{itemize}

We now introduce a condition that implies \Cref{th:window-sort:error-analysis}: Throughout the execution of \WindowSort we would like every element $x$ to satisfy the following condition:
\begin{quote}\begin{itemize}
		\item[{\Large\INV}]  For window size $w$, the dislocation of $x$ is at most $w$.
	\end{itemize}
\end{quote}

We also introduce two further conditions, which essentially relax the requirement that all elements satisfy \INV.  
The first condition justifies the first step of our algorithm, while the second condition restricts the range of elements that get compared with $x$ in some iteration:
\begin{quote}\begin{itemize}
		\item[{\Large\first}]  For window size $w$, element $x$ is larger (smaller) than all the elements lying apart by more than $2w$ positions left (right) of $x$'s original position.
	\end{itemize}
\end{quote}
\begin{quote}\begin{itemize}
		\item[{\Large\second}]  For window size $w$, $x$ and its left $2w$ and right $2w$ adjacent elements satisfy condition~\INV.
	\end{itemize}
\end{quote}

Note that if  \INV\ holds for all elements, then \first\ and \second\ also hold for all elements.
For elements that satisfy both \first\ and \second , the computed rank is close to the true rank if there are few errors in the comparisons:

\begin{lemma}\label{lem-analyzable}
	For every window size $w$, if an element $x$ satisfies satisfy both \first\ and \second, then the absolute difference between the computed rank and its true rank is bounded by 
	\[
	|computed\_rank(x) - x| \leq |\err{x}{w}| \enspace .
	\]
\end{lemma}
\begin{proof}
	This follows immediately from condition \first.
\end{proof}

We now consider the difference between the computed rank and the computed position of an element, which we define as the \emph{offset} of this element. Afterwards, we consider the  difference between the original position and the computed position of an element.

\begin{fact}\label{fact}
	Observe that by Step~\ref{alg:computed_rank} of the algorithm it holds that,
	for every permutation $\sigma$, every window size $w$ and every element $x$, the difference between $\sigma(x)$ and the computed rank of $x$ is at most $2w$, 
	\[
	|computed\_rank(x) - \sigma(x)| \leq 2w \enspace .
	\]
\end{fact}

\begin{lemma}\label{lem-difference}
	For any permutation of $n$ elements and for each element $x$, if the difference between the computed rank and $x$ is at most $m$ for every element,
	then the difference between the computed position and $x$ is at most $2m$ for every element.
\end{lemma}
The proof of this lemma is analogue to the proof of Lemma~\ref{lem-diff} below.

\begin{lemma}\label{lem-diff}
	For every permutation $\sigma$ and window size $w$, the offset of every element $x$ is at most $2w$,  
	\[
	|computed\_rank(x) - \sigma'(x)| \leq 2w \enspace .
	\]
\end{lemma}
\begin{proof}
Let the computed rank of $x$ be $k$. The computed position $\sigma'(x)$ is larger than the number of elements with computed rank smaller than $k$,
and at most the number of elements with computed ranks at most $k$.
By Fact \ref{fact},
every element $y$ with $\sigma(y)<k-2w$ has a computed rank smaller than $k$,
and every element $y$ with $\sigma(y)>k+2w$ has a computed rank larger than $k$.  
\end{proof}

\begin{lemma}\label{lem-offset}
	Consider a generic iteration of the algorithm with  permutation $\sigma$ and a window size $w$.
	In this iteration, the position of each element changes by at most $4w$. Moreover, the position of each element changes by at most $8w$ until the algorithm terminates. 
\end{lemma}
\begin{proof}
By Fact \ref{fact}, Lemma~\ref{lem-diff} and triangle inequality,
$$\lvert \sigma(x) - \sigma'(x) \rvert \leq \lvert computed\_rank(x) - \sigma(x) \rvert + \lvert computed\_rank(x) - \sigma'(x) \rvert \leq 2w + 2w =4w.$$
Since $w$ is halved after every iteration, the final difference is at most $\sum_{i=0}^{\infty}\frac{4w}{2^i}=8w$.
\end{proof}

Finally, we conclude \Cref{th:window-sort:error-analysis} and show that the  dislocation of an element is small if the number of errors is small:

\begin{proof}[Proof of \Cref{th:window-sort:error-analysis}]
	Consider an iteration of the algorithm with current window size $w$. 
	We show that, if \INV{} holds for all elements  in the current iteration, then \eqref{eq:error-condition} implies that \INV\ also holds for all elements in the next iteration, i.e., when the window size becomes $w/2$. 
	In order for \INV{} to hold for the next iteration, 
	the computed \emph{position} of each element should differ from the true rank by at most ${w}/{2}$,
	\[
	|\sigma'(x) - x| \leq w/2 \enspace .
	\]
	By Lemma \ref{lem-difference},
	it is sufficient to require that the computed \emph{rank} of each element differs from its true rank by at most ${w}/{4}$, 
	\[
	|computed\_rank(x) - x|  \leq w/4
	\]
	By Lemma~\ref{lem-analyzable}, the above inequality follows from the hypothesis $|\err{x}{w}| \leq w/4$.
	
	We have thus shown that after the iteration with window size $2w^*$, all elements have dislocation at most $w^*$. By Lemma \ref{lem-offset}, the subsequent iterations will move each element by at most $8w^*$ positions.
\end{proof}

\begin{remark}
	If we care only about the maximum dislocation, then we could obtain a better bound of $w$ by simply stopping the algorithm at the iteration where the window size is $w$ (for a $w$ which guarantees the condition above with high probability). In order to bound also the total dislocation, we let the algorithm continue all the way until window size $w=1$. This will allow us to show that the total dislocation is linear in expectation.
\end{remark}

\section{Maximum Dislocation}\label{sec-max}
In this section we give a bound on the maximum dislocation of an element after running \WindowSort\ on $n$ elements. We prove that it is a function of $n$ {and} of the probability $p$ that a single comparison fails. Our main result is the following:
\begin{theorem}\label{thm-maximum-dislocation}
	For a set of $n$ elements, with probability $1-{1}/{n}$, the maximum dislocation after running \WindowSort\
	is $9\cdot f(p)\cdot \log n$ where
	\[ f(p) =
	 \begin{cases}~
          \frac{400p}{(1-32p)^2} &  \quad  \text{for} \quad 1 /64 < p < 1/32,\\~
   	 \frac{4}{\ln \left(\frac{1}{32p}\right)-\left(1-32p\right)} &  \quad  \text{for} \quad 1/192 < p \leq 1/64,\\~
         6 & \quad  \text{for} \quad  p \leq 1/192.
	  \end{cases}
	\]
\end{theorem}

It is enough to prove that the condition in  Theorem~\ref{th:window-sort:error-analysis} holds for all  $w\geq 2f(p)\log  n$ with probability at least $1 - 1/n$.

\begin{lemma}\label{le:error-condition}
For every fixed element $x$ and for every fixed window size $w\geq 2f(p)\log n$, the probability that
	\begin{equation}
	|\err{x}{w}| > w/4 \label{eq:error-condition-negated}
	\end{equation}
	is at most $1/n^3$.
\end{lemma}
By the union bound, the probability that \eqref{eq:error-condition-negated} holds for some $x$ and for some $w$ is at most $1/n$. That is, the condition of Theorem~\ref{th:window-sort:error-analysis} holds with probability at least $1 - 1/n$ for all $w\geq 2w^\star=2f(p)\log n$, which then implies Theorem~\ref{thm-maximum-dislocation}. 

\subsection{Proof of Lemma~\ref{le:error-condition}}
Since each comparison fails with probability $p$ independently of the other comparisons, the probability that the event in \eqref{eq:error-condition-negated} happens is equal to the probability that at least ${w}/{4}$ errors occur in $8 w$ comparisons. We denote such probability as $\Prob(w)$, and show that 
$
\Prob(w) \leq 1/n^3 \enspace.
$

We will make use of the following standard Chernoff Bounds (see for instance in \cite{MitzenmacherU05}):
\begin{theorem}[Chernoff Bounds]\label{thm-chernoff} 
Let $X_1,\cdots,X_n$ be independent Poisson trials with $\Pr(X_i)=p_i$. Let $X=\sum_{i=1}^nX_i$ and $\mu=\E[X]$.  Then the following bounds hold:
        \begin{align}
       		(i)& \text{ For $0< \delta < 1$,}&&
       		\Pr(X\geq (1+\delta )\mu )\leq e^{-\frac{\mu\delta^2}{3}},\\
            (ii)&  \text{ For any  $\delta >0$,}&&
            \Pr(X\geq (1+\delta )\mu )< \bigg(\frac{e^\delta }{(1+\delta)^{(1+\delta )}}\bigg)^\mu,\\
            (iii)& \text{ For $R\geq 6\mu$,}&& \Pr(X\geq R )\leq 2^{-R}.
        \end{align}
\end{theorem}

\begin{lemma}\label{lem-pr-fail}
	The probability $\Prob(w)$ (at least ${w}/{4}$ errors occur in $8 w$ comparisons) satisfies
\[ \Prob(w) \leq
  \begin{cases}~
          e^{-\frac{w\left(1-32p\right)^2}{384p}} & \quad  \text{for} \quad1 /64 < p < 1/32,\\[6pt]
    \bigg( \frac{e^{\frac{1-32p}{32p}}}{\left(\frac{1}{32p}\right)^{\frac{1}{32p}}}\bigg)^{8 w p}&  \quad  \text{for}  \quad 1/192 < p \leq 1/64, \\[12pt]
         2^{-\frac{w}{4}}   &  \quad  \text{for} \quad  p \leq 1/192.
  \end{cases}
\]
\end{lemma}
\begin{proof}
Let the random variable $X$ denote the number of errors in the outcome of $8w$ comparisons. Clearly, $\E[X]=8w p$, and
\begin{align*}
	\Prob (w)&=\Pr \left[ X\geq \frac{w}{4}\right]  
	=\Pr\left[X\geq \frac{\E[X]}{32p}\right] 
	= \Pr\left[X\geq \left(1+\frac{1-32p}{32p}\right) \E[X]\right].
\end{align*} 
Let $\delta=\frac{1-32p}{32p}$. 
If $1/64 < p <1/32$, then $0< \delta< 1$, and by Theorem~\ref{thm-chernoff}, case~($i$), we have 
\[\Prob(w)\leq e^{-\frac{1}{3}\delta^2\mu}\leq   e^{-\frac{w\cdot\left(1-32p\right)^2}{384p}}.\]
Similarly, if $p \leq 1/64$, then $\delta \geq1$, and by Theorem~\ref{thm-chernoff}, case ($ii$), we have
\[\Prob(w)< \left(\frac{e^\delta}{(1+\delta)^{(1+\delta)}}\right)^\mu \leq  \bigg( \frac{{\scriptstyle e^{\frac{1-32p}{32p}}}}{{\scriptstyle \left(\frac{1}{32p}\right)}^{\frac{1}{32p}}}\bigg)^{8 w p}\ .\]
If $p\leq 1/192$, then ${w}/{4}\geq 48 w p =6\, \E[X]$, and by Theorem~\ref{thm-chernoff} case (iii),
$\Prob(w)\leq 2^{-\frac{w}{4}}.$
\end{proof}

\begin{lemma}\label{lemma-pr-cubic}
If $w\geq 2\, f(p)\log n$, with $n\geq 1$ and $f(p)$ as in Theorem~\ref{thm-maximum-dislocation}, then $\Prob(w)\leq 1/n^{3}$.
\end{lemma}
\begin{proof}
We show the first case, the other two are similar. 
If $1/64 < p <1/32$, by Lemma~\ref{lem-pr-fail},
\[\Prob(w)\leq e^{-\frac{w\cdot\left(1-32p\right)^2}{384p}} \leq e^{-\frac{800}{384}\log n}\leq e^{-3\lg n}\leq 1/n^{3}.\qedhere\]
\end{proof}

\section{Total Dislocation}\label{sec-total}

In this section, we prove that \WindowSort\ orders $n$ elements such that their total dislocation is linear in $n$ times a factor which depends only on $p$: 

\begin{theorem}\label{thm:TD}
	For a set of $n$ elements, the expected total dislocation after running \WindowSort\ is at most $n \cdot 60 \, f(p) \log f(p)$.
\end{theorem}

The key idea is to show that for an element $x$, only $O(w)$ elements adjacent to its (true) rank matter in  all upcoming iterations. If this holds, it is sufficient to keep the following \emph{weak invariant} for an element $x$ throughout all iterations: 
\begin{quote}
	\begin{itemize}
		\item[{\Large\weakINV}] This invariant consists of three conditions that have to be satisfied:
		\begin{itemize}
			\item[~]  
			\begin{enumerate}[(a)]
				\item $x$ satisfies condition \INV.
				\item All elements with original position in $[x-12w, x+12w]$ satisfy condition \INV.
				\item All elements with original position in $[x-10w, x+10w]$ satisfy condition \first.
			\end{enumerate}
		\end{itemize}
		
	\end{itemize}
\end{quote}
Note that if $x$ satisfies \weakINV, all elements lying in $[x-10w, x+10w]$ satisfy both \first\ and \second. 

The rest of this section is structured as follows: First we derive several properties of the weak invariant, then we prove an $n\log\log n$ bound on the expected total dislocation, and finally we extend the proof to achieve the claimed linear bound. 

\subsection{Properties of the Weak Invariant}\label{sub-weak} 

We start with the key property of the weak invariant \weakINV\ for some element $x$.

\begin{lemma}\label{lem-weak}
Let $\sigma$ be the permutation of $n$ elements and $w$ be the window size of some iteration in \WindowSort. If the weak invariant \emph{\weakINV}\ holds for an element $x$ in $\sigma$ and the computed rank of every element $y$ with $\sigma(y)\in[x-10w, x+10w]$ differs from $y$ by at most ${w}/{4}$, 
then \emph{\weakINV}\ still holds for $x$ in the permutation $\sigma'$ of next iteration with window size ${w}/{2}$.
\end{lemma}
\begin{proof}
Consider the set $X$ of all elements $y$ with $computed\_rank(y)\in[x-8w, x+8w]$.
Their computed ranks differ from their original positions by at most $2w$.
Thus, all these elements are in the set $Y\supseteq X$ of all elements whose original positions are in $[x-10w, x+10w]$.
By the assumption of the lemma, for each element $y\in Y$, $\lvert computed\_rank(y) - y \rvert \leq {w}/{4}$.
Using the same reasoning as in the proof of Lemma \ref{lem-diff}, we  conclude that 
\begin{aquote}{A}
	for each element $y\in X$,  $\lvert \sigma'(y)-y \rvert \leq{w}/{2}$.
\end{aquote} 
Consider the set $Z$ of all elements $y$ with $\sigma'(y)\in[x-6w,x+6w]$.
By Lemma \ref{lem-diff}, their computed ranks lie in $[x-8w, x+8w]$, thus $Z\subseteq X$, and by (A), $\lvert \sigma'(y) - y \rvert\leq {w}/{2}$ for each $y\in Z$.
Thus, the second condition of \weakINV\ holds for the next iteration.

We continue with the third condition. Consider the set $T\subseteq Z$ of all elements $y$ with $\sigma'(y)\in [x-5w,x+5w]$. By the assumptions of the lemma, $y\in[x-5w-{w}/{2},x+5w+{w}/{2}]$ and $\sigma(y)\in[x-5w-{3w}/{2},x+5w+{3w}/{2}]$ for all $y\in T$.
It is sufficient to show that every element in $ T$ is larger (or smaller) than all elements whose computed positions are smaller than $x-6w$ (or larger than $x+6w$), the rest follows from the second condition. We show the former case, the latter is symmetric.
We distinguish three subcases: elements $y\in T$ with $\sigma'(y)<x-6w$ and with $\sigma(y)$ (i) smaller than $x-12$, (ii) between $x-12$ and $x-10w-1$, or (iii) between $x-10w$ and $x-4w-1$.
\begin{enumerate}[(i)]
	\item This case follows immediately from the third condition of \weakINV.
	\item This case follows immediately from the second condition of \weakINV.
	\item By the assumption of our lemma, $\lvert computed\_rank(y) - y \rvert \leq {w}/{4}$.
			Thus, if the computed rank $r$ of such an element $y$ is smaller than $x-6w$, then $y<x-6w+{w}/{4}$.
			Otherwise, if $r\geq kx-6w$, then by (A), $\lvert y-\sigma'(y) \rvert \leq {w}/{2}$. Thus, $y<x-6w+{w}/{2}$.
\end{enumerate}

Since we assume \weakINV\ for $x$, $\sigma(x) \in [x-w, x+w]$ and $computed\_rank(x)\in[x-3w,x+3w]$. By Lemma \ref{lem-diff}, $\sigma'(x)\in [k-5w, k+5w]$, and thus $x\in Z$, which implies that the first condition of \weakINV\ will still be satisfied for $x$ for the next iteration. This concludes the proof.
\end{proof}

Next, we adopt Lemma~\ref{lemma-pr-cubic} to analyze the probability of keeping the weak invariant for an element $x$ and an arbitrary window size through several iteration of \WindowSort.

\begin{lemma}\label{lem-weak-pr}
	Consider an iteration of \WindowSort\ on a permutation $\sigma$ on $n$ elements such that the window size is $w\geq 2\, f(p)  \log w$, where $f(p)$ is defined as in Theorem \ref{thm-maximum-dislocation}.
If the weak invariant \emph{\weakINV}\ for an element $x$ holds,
then with probability at least $1-{42}/{w^2}$, \emph{\weakINV}\ still holds for $x$  when the window size is $f(p) \log w$ (after some iterations of \WindowSort). 
\end{lemma}
\begin{proof}
By Lemma~\ref{lem-weak}, the probability that \weakINV\ fails for $x$ before the next iteration is $(20w+1)\cdot \Prob(w)$.
Let $r=\log(\frac{w}{2\, f(p)\\log w})$, then the probability that \weakINV\ fails for $x$ during the iterations from window size $w$ to window size $f(p)\cdot \log w$ is
\begin{align*}
\sum_{i=0}^{r}\left(\frac{20w}{2^i}+1\right)\cdot\Prob\left(\frac{w}{2^i}\right) &\leq 
  \sum_{i=0}^{r}\left(\frac{21w}{2^i}\right)\cdot\Prob\left(2\, f(p) \log w\right) \leq 42w \cdot\Prob\left(2\, f(p) \log w\right),
\end{align*}
where the first inequality is by fact that $\Prob(w)$ increases when $w$ decreases.
By Lemma~\ref{lemma-pr-cubic}, $\Prob(2\, f(p) \log w)\leq {1}/{w^3}$, leading to the statement.
\end{proof}

\subsection{Double Logarithmic Factor (Main Idea)}\label{sub-d-log}
Given that \WindowSort\ guarantees maximum dislocation at most $9\, f(p) \log n$ with probability at least $(1-{1}/{n})$ (Theorem~\ref{thm-maximum-dislocation}), this trivially implies that the expected total dislocation is at most $O(f(p)\log n)$. More precisely, the expected dislocation is at most
\begin{align*}
\left({1}/{n}\right)\cdot n \cdot n + \left(1-{1}/{n}\right)\cdot n \cdot 9 \, f(p) \log n\leq n\cdot (1+9\,f(p) \log n)\, ,
\end{align*}
since a fraction $1/n$ of the elements is dislocated by at most $n$, while the others are dislocated by at most $9\cdot f (p)  \log n$.

We next describe how to improve this to  $O(f(p)\log\log n)$ by considering in the analysis \emph{two phases} during the execution of the algorithm:
\begin{itemize}
	\item \textbf{Phase 1:} The first phase consists of the iterations up to window size $w=f(p)\log n$. With probability at least $(1-1/n)$ all elements satisfy \INV\ during this phase.
	\item \textbf{Phase 2:} The second phase consists of the executions up to window size $w' = f(p)\log w$. If all elements satisfied \INV\ at the end of the previous phase, then the probability that a fixed element violates \weakINV\ during this second phase is at most $42/w^2$. 
\end{itemize}
More precisely, by \Cref{thm-maximum-dislocation} and the proof of \Cref{th:window-sort:error-analysis}, the
probability that \INV\ holds for all elements when the window size is $f(p) \log n$ is at least $(1-{1}/{n})$. We thus  restart our analysis with $w=f(p) \log n$ and the corresponding permutation $\sigma$.
Assume an element $x$ satisfies \weakINV. 
By Lemma~\ref{lem-weak-pr},
the probability that \weakINV\ fails for $x$ before the window size is $f(p) \log w$ is at most ${42}/{w^2}$. 
By Lemma~\ref{lem-offset},
an element moves by at most $8w$ positions from its original position, which is at most $w$ apart from its true rank. Therefore, the expected dislocation of an element $x$ is at most
\begin{equation*}
\left({1}/{n}\right)\cdot n+  {42}/{w^2}\cdot 9w +  9 \, f(p) \log w = O(1) + 9 \, f(p) \log (f(p) \log n)\ , 
\end{equation*}
where the equality holds for sufficiently large $n$  because $w=f(p)\log n$.

\subsection{Linear Dislocation (Proof of Theorem \ref{thm:TD})}\label{sub-logarithmic-dislocation}
In this section, we apply a simple idea to decrease the upper bound on the expected total dislocation after running \WindowSort\ on $n$ elements to $60 \, f(p) \log f(p)$.
We recurse the analysis from the previous Section \ref{sub-d-log} for several \textit{phases}:
Roughly speaking, an \textit{iteration} in \WindowSort\ halves the window size,
a \textit{phase} of iterations logarithmizes the window size.

\begin{itemize}
	\item \textbf{Phase 1}: Iterations until the window size  is $f(p)\log n$.
	\item \textbf{Phase 2:} Subsequent iterations until the window size  is $f(p)\cdot\log( f(p)\log n)$.
	\item \textbf{Phase 3:} Subsequent iterations until the window size is $ f(p)\cdot\log(f(p)\cdot \log (f(p)\log n))$.
	\item \dots
\end{itemize}

We bound the expected dislocation of an element $x$,
and let $w_i$ denote the window size after the $i$-th phase. 
We have $w_0=n$, $w_1=f(p) \log n$, $w_2=f(p) \log ( f(p) \log n)$, and 
\begin{equation}\label{eq:phase}
w_{i+1}= f(p)  \log w_i,
\end{equation} if $i\geq 1$ and $w_i \geq 2\, f(p)  \log w_i$.
Any further phase would just consist of a single iteration. In the remaining of this section, we only consider phases $i$ for which \Cref{eq:phase} is true, and we call them the \emph{valid phases}.

By Lemma~\ref{lem-weak-pr},
if the weak invariant \weakINV\ holds for $x$ and window size $w_{i-1}$, 
the probability that it still holds for window size $w_{i}$
is at least $1-{42}/{w_{i-1}^2}$.
Similarly to the analysis in the Section~\ref{sub-d-log}, we get that a valid phase $i\geq 1$ contributes to the expected dislocation of $x$ by 
\begin{eqnarray}\label{eq:disl}
 {42}/{w_{i-1}^2}\cdot 9w_{i-1}={378}/{w_{i-1}}\, .
\end{eqnarray}
If we stop our analysis after $c$ valid phases,
then by \eqref{eq:disl} and Lemma \ref{lem-offset}, the expected dislocation of any element $x$ is at most
\begin{equation}\label{eq:wcup}
\sum_{i=0}^{c-1}{378}/{w_i} + 9w_c \leq {378}/{w_c} + 9w_c \, .
\end{equation}
The inequality holds since $\frac{w_{i-1}}{w_{i}}\leq 2$ for $1<i<c$.
We next define $c$ such that phase $c$ is valid and $w_c$ only depends on $f(p)$.
The term $\frac{w_{i-1}}{ f(p)\log w_{i-1}} \geq 2$ holds for every valid phase $i$ and decreases with increasing $i$.
For instance for $w=6f(p)\log f(p)$:
\begin{equation*}
\frac{w}{f(p)\log w} = \frac{6\,f(p)\log f(p)}{f(p)\log (6\,f(p)\log f(p))}
\geq \frac{6\,\log f(p)}{3\,\log f(p)}\geq 2\, .
\end{equation*}
Therefore, if we choose $c$ such that $w_{c-1}\geq 6\,f(p)\log f(p)> w_c$, we can use that $f(p)\geq 6$ and upper bound $w_c$
by
\begin{equation}\label{eq:wcdown}
w_c = f(p)\log w_{c-1} \geq f(p)\log (6\,f(p)\log f(p)) \geq 6 \log (36\log 6) \geq 39\ .
\end{equation}

\Cref{eq:wcup,eq:wcdown} and Lemma~\ref{lem-offset} imply the following:
\begin{lemma}
	The expected dislocation of each element $x$ after running \WindowSort\ is at most
	\begin{equation*}
	{378}/{w_c}+9w_c <  10+9w_c\leq 10 w_c \leq 60 \, f(p) \log f(p) \enspace.
	\end{equation*}
\end{lemma}
 This immediately implies Theorem~\ref{thm:TD}.

\section{Extension}\label{sec-extension}

The reason why we require the error probability $p$ to be smaller than $1/32$
is to analyze the probability that at most $w/4$ errors occur in $8w$ comparisons, for $w\geq 1$. This bound on the number of errors appears since we halve the window size in every iteration.
If we let the window size shrink by another rate $1/2<\alpha<1$,
the limit of $p$ will also change:

First, the running time of the adapted \WindowSort\ will become $O(\frac{1}{1-\alpha}n^2)$.
Second, for any permutation $\sigma$ and window size $w$, in order to maintain condition~\INV\ for an element $x$,
its computed position should differ from $x$ by at most $\alpha w$,
and thus $computed\_rank(x)$ should differ from $x$ by at most $\alpha w/2$. 

Our new issue is thus the probability that at most $\alpha w/2$ errors occur in $8w$ comparisons:
Since the expected number of errors is $8wp$, we have 
\[\frac{\alpha w}{2} = \frac{\alpha}{16p}\cdot 8wp=(1+\frac{\alpha-16p}{16p}) \cdot 8wp\, ,\]
and by the reasoning of Lemma~\ref{lem-pr-fail},
we have $\frac{\alpha-16p}{16p}> 0$, thus $p<{\alpha}/{16}$.
(Note that $f(p)$  should change accordingly.)

Finally, the number of windows for the weak invariant should also change accordingly. Let $m$ be the number of windows that matter for the weak invariant ($m=12$ when $\alpha=1/2$).
According to the analysis in Section~\ref{sub-weak}, we have $m-6\geq \alpha m$, implying that $m\geq \frac{6}{1-\alpha}$.
Of course, the constant inside the linear expected total dislocation will also change accordingly. 

\begin{theorem}
For an error probability $p<{\alpha}/{16}$, where $1/2< \alpha <1$,
 modified \WindowSort\ on $n$ elements takes $O(\frac{1}{1-\alpha}n^2)$ time,  has maximum dislocation $9\, g(p,\alpha) \log n$ with probability $1-1/n$,
and  expected total dislocation  $n \cdot (9+\frac{2}{1-\alpha})\cdot 6 \, g(p,\alpha) \log g(p,\alpha)$,
where
\[ g(p,\alpha) =
	 \begin{cases}
          \frac{100p}{(\alpha-16p)^2} &  \quad  \text{for} \quad \alpha /32 < p < \alpha/16\ ,\\
   	 \frac{4}{(\ln (\alpha/16p))-(\alpha-16p)} &  \quad  \text{for} \quad \alpha /96 < p \leq \alpha /32\ ,\\
         6 & \quad  \text{for} \quad  p \leq \alpha /96\ .
	  \end{cases}
	\]
\end{theorem}

\section{A lower bound on the maximum dislocation}\label{sec:lowerbound}

In this section we prove a lower bound on both the maximum and the average dislocation that can be achieved w.h.p. by any sorting algorithm. 

Let $S = \{1, 2, \dots, n\}$ be the set of elements to be sorted.
We can think of an instance of our sorting problem as a pair $\langle \pi, C \rangle$ where $\pi$ is a permutation of $S$ that represents the order of the element in the input collection and $C = (c_{i,j} )_{i,j}$ is a $n \times n$ matrix that encodes the result of the comparisons as seen by the algorithm. More precisely, $c_{i,j}$ is \text{``$<$''} if $i$ is reported to be smaller than $j$ when comparing $i$ and $j$ and $c_{i,j}=\text{``$<$''}$ otherwise.
Notice that $c_{i,j}=\text{``$<$''}$ iff $c_{j,i}=\text{``$>$''}$ and hence in what follows we will only define $c_{i,j}$ for $i < j$.

The following lemma -- whose proof is moved to \Cref{app:proofs} -- is a key ingredient in our lower bounds:
\begin{lemma}
\label{lemma:swap_probability}
Let $x,y \in S$ with $x < y$. Let $A$ be any (possibly randomized) algorithm. On a random instance, the probability that $A$ returns a permutation in which
elements $x$ and $y$ appear the wrong relative order is at least $\frac{1}{2} \left( \frac{p}{1-p} \right)^{2(y-x)-1}$. 
\end{lemma}

As a first consequence of the previous lemma, we obtain the following:
\begin{theorem}
	\label{thm:lb_max_dislocation}
	No (possibly randomized) algorithm can achieve maximum dislocation $o(\log n)$  with high probability.
\end{theorem}
\begin{proof}
	By Lemma~\ref{lemma:swap_probability}, any algorithm, when invoked on a random instance, must return a permutation $\rho$ in which elements $1$ and $h= \big\lfloor  \frac{\log n }{2 \log {1-p}/{p}} \big\rfloor$ appear in the wrong order with a probability larger than $\frac{1}{n}$.
When this happens, at least one of the following two conditions holds: (i) the position of element $1$ in $\rho $ is at least $\lceil \frac{h}{2} \rceil$; or (ii) the position of element $h$ in $\rho$ is at most $\lfloor \frac{h}{2} \rfloor$.
In any case, the maximum dislocation must be at least $\frac{h}{2} - 1 = \Omega(\log n)$.
\end{proof}

Finally, we are also able to prove a lower bound to the total dislocation (whose proof is given in \Cref{app:proofs}).
\begin{theorem}
	\label{thm:lb_tot_dislocation}
	No (possibly randomized) algorithm can achieve expected total dislocation $o(n)$.
\end{theorem}

\bibliographystyle{abbrv}
\bibliography{references}

\appendix
\section{Experimental Results}\label{sec-exp}

In this section, we discuss some experimental results on the performance of \WindowSort for increasing error probability $p$ (see Tables~\ref{tb-average} and \ref{tb-maximum}).
Our results suggest that in practice, the probability of error can be much higher than in our theoretical analysis (i.e. $p<1/32$). In these experiments we measure the average dislocation (which gives us an estimate of the expected total dislocation) and the maximum dislocation among all elements.

Note that in all the experiments, 
\WindowSort performs significantly better than the theoretical guarantees (see Theorem~\ref{thm-maximum-dislocation} and \ref{thm:TD} for the corresponding values of $p$).
Also, the experiments suggest that the expected total dislocation  is $O(n)$ for $p<1/5$, since the average dislocation seems to not increase with $n$ in these cases. 
Analogously, the maximum dislocation seems to be $O(\log n)$ for $p<1/4$.

We consider five different values for the input-size $n$ and ten different values for the error probability $p$. 
Each setting consists of one-hundred instances. We use the error probability to generate a comparison table among the $n$ elements for simulating recurrent comparison errors.
Moreover, we set $\alpha$ to be $1/2$ as in the standard version of \WindowSort.

\begin{table}[h]
	\begin{tabular}{|r|r|r|r|r|r|r|r|r|r|r|}
		\hline
		\backslashbox{$n$}{$p$}  & 1/3  & 1/4  & 1/5  & 1/8  & 1/12  & 1/16  & 1/20  & 1/24  & 1/28  & 1/32 \\
		\hline
		1024  & 14.160  & 4.873  & 2.870  & 1.377  & 0.881  & 0.670  & 0.536  & 0.454  & 0.390  & 0.346 \\
		\hline
		2048  & 15.993  & 4.984  & 2.884  & 1.397  & 0.895  & 0.674  & 0.541  & 0.464  & 0.394  & 0.348 \\
		\hline
		4096  & 17.494  & 5.075  & 2.904  & 1.390  & 0.893  & 0.673  & 0.545  & 0.460  & 0.398  & 0.351 \\
		\hline
		8192  & 19.030  & 5.105  & 2.898  & 1.395  & 0.894  & 0.675  & 0.545  & 0.460  & 0.397  & 0.351 \\
		\hline
		16384  & 20.377  & 5.123  & 2.902  & 1.390  & 0.892  & 0.673  & 0.545  & 0.460  & 0.398  & 0.349 \\
		\hline
	\end{tabular}
	\caption{The average dislocation of one element.}\label{tb-average}
\end{table}

\begin{table}[h]
	\begin{tabular}{|r|r|r|r|r|r|r|r|r|r|r|}
		\hline
		\backslashbox{$n$}{$p$}  & 1/3  & 1/4  & 1/5  & 1/8  & 1/12  & 1/16  & 1/20  & 1/24  & 1/28  & 1/32 \\
		\hline
		1024  & 15.600  & 5.400  & 3.900  & 2.100  & 1.200  & 0.900  & 0.900  & 0.600  & 0.600  & 0.600 \\
		\hline
		2048  & 17.000  & 4.909  & 2.818  & 1.545  & 1.091  & 0.818  & 0.818  & 0.636  & 0.636  & 0.636 \\
		\hline
		4096  & 17.917  & 5.333  & 2.833  & 1.500  & 1.083  & 0.833  & 0.917  & 0.667  & 0.583  & 0.667 \\
		\hline
		8192  & 18.923  & 6.462  & 3.077  & 1.769  & 1.154  & 0.769  & 0.692  & 0.769  & 0.538  & 0.538 \\
		\hline
		16384  & 22.714  & 5.071  & 3.929  & 1.786  & 1.071  & 0.857  & 0.643  & 0.643  & 0.571  & 0.571 \\
		\hline
	\end{tabular}
	\caption{The maximum dislocation divided by $\log n$.}\label{tb-maximum}
\end{table}
\section{Omitted Proofs}\label{app:proofs}

\subsection*{Proof of Lemma~\ref{lemma:swap_probability}}
	We prove that, for any running time $t>0$, no algorithm $A$ can compute, within $t$ steps, a sequence in which $x$ and $y$ are in the correct relative order with a probability larger than $1-\frac{1}{2} \left( \frac{p}{1-p} \right)^{2(y-x)-1}$.

	First of all, let us focus a deterministic version of algorithm $A$ by fixing a sequence $\lambda \in \{0, 1\}^t$ of random bits that can be used be the algorithm. We call the resulting algorithm $A_\lambda$ and we let $p_\lambda$ denote the probability of generating the sequence $\lambda$ of random bits (i.e., $2^{-t}$ if the random bits come from a fair coin).\footnote{Notice that $A_\lambda$ can use less than $t$ random bits, but not more (due to the limit on its running time).}
	Notice that $A$ might already be a deterministic algorithm, in this case $A = A_{\lambda}$ for every $\lambda \in \{0, 1\}^t$.

	Consider an instance $I = \langle \pi, C \rangle$, and let $\phi_I$ be the permutation of the elements in $S$ returned by $A_\lambda(I)$, i.e., $\phi_I(i)=j$ if the element $i \in S$ is the $j$-th element of the sequence returned by $A_\lambda$.
  We define a new instance $I_{\pi,C} = \langle \pi', C' \rangle$ by ``swapping'' elements $x$ and $y$ of $I$ along with the results of all their comparisons. More formally we define: $\pi'(i)=\pi(i)$ for all $i \in S \setminus \{x, y\}$, $\pi'(x)=\pi(y)$, and $\pi'(y)=\pi(x)$; We define the comparison matrix $C' = (c'_{i,j})_{i,j}$ accordingly, i.e.:%
\indent\begin{multicols}{2}  
 \begin{itemize}
	\item $c'_{i,j} = c_{i,j}$ if $i,j \in S \setminus \{x, y\}$ and $j>i$,
  	\item $c'_{i,x} = c_{i,y}$ if $i < x$,
 	\item $c'_{x,j} = c_{y,j}$ if $j>x$ and $j \neq y$,
    \item $c'_{i,y} = c_{i,x}$ if $i < y$ and $i \neq x$,
    \item $c'_{y,j} = c_{x,j}$ if $j > y$,
	\item $c'_{x,y} = c_{y, x}$.
  \end{itemize}
\end{multicols}%
\noindent Letting $R = ( r_{i,j} )_{i,j}$ be a random variable over the space of all comparison matrices, we get:
\begin{multline*}
	\Pr(R = C') = \prod_{i<j} \Pr(r_{i,j} = c'_{i,j}) \\
     = \prod_{\substack{i<j \\ i,j \not\in \{x, y\} }} \Pr(r_{i,j} = c_{i,j}) \cdot 
    \prod_{i<x} \Pr(r_{i,x} = c_{i,y}) \cdot
	\prod_{x<j<y} \Pr(r_{x,j} = c_{y,j}) \cdot
    \prod_{j>y} \Pr(r_{x,j} = c_{y,j}) \cdot \\
     \qquad \prod_{i<x} \Pr(r_{i,y} = c_{i,x}) \cdot 
    \prod_{x<i<y} \Pr(r_{i,y} = c_{i,x}) \cdot
    \prod_{j>y} \Pr(r_{y,j} = c_{x,j}) \cdot
    \Pr(r_{x,y} = c_{y,x} ) \\
     = \prod_{\substack{i<j \\ i,j \not\in \{x, y\} }} \Pr(r_{i,j} = c_{i,j}) \cdot 
    \prod_{i<x} \Pr(r_{i,y} = c_{i,y}) \cdot
    \prod_{x<j<y} (1-\Pr(r_{x,j} = c_{j,y})) \cdot
    \prod_{j>y} \Pr(r_{y,j} = c_{y,j}) \cdot \\
     \qquad \prod_{i<x} \Pr(r_{i,x} = c_{i,x}) \cdot 
    \prod_{x<i<y} (1-\Pr(r_{i,y} = c_{x,i})) \cdot
    \prod_{j>y} \Pr(r_{x,j} = c_{x,j}) \cdot
	(1- \Pr(r_{x,y} = c_{x,y} )) \\
     =  \prod_{i<j} \Pr(r_{i,j} = c_{i,j}) \bigg/ \left( \prod_{x<j<y} \Pr(r_{x,j} = c_{x,j}) \cdot  \prod_{x<i<y} \Pr(r_{i,y} = c_{i,y}) \cdot \Pr(r_{x,y} = c_{x,y} ) \right) \cdot \\
     \quad \prod_{x<j<y} (1-\Pr(r_{x,j} = c_{j,y})) \cdot
    \prod_{x<i<y} (1-\Pr(r_{i,y} = c_{x,i})) \cdot
    (1- \Pr(r_{x,y} = c_{x,y} )) \\
     \ge P(R=C) \cdot \left(\frac{1}{1-p}\right)^{2(y-x)-1} \cdot p^{2(y-x)-1} 
    = \Pr(R=C) \cdot \left(\frac{p}{1-p}\right)^{2(y-x)-1}
\end{multline*}
\noindent where we used the fact that $r_{i,j}$ and $r_{i',j'}$ with $i < j$ and $i' < j'$ have the same probability distribution.

Let $\widetilde{\pi}$ be a random variable whose value is chosen u.a.r. over all the permutations of $S$.
Let $\Pr(I)$ be the probability that instance $I=\langle \pi, C \rangle$ is to be solved and notice that $\Pr(I) = \Pr(\widetilde{\pi} = \pi) \Pr(R = C)$.
It follows from the above discussion that $\Pr(I_{\pi,c}) \ge \Pr( I ) \cdot \left( \frac{p}{1-p} \right)^{2(y-x)-1}$.

Let $X_I$ be an indicator random variable that is $1$ iff algorithm $A_\lambda$ on instance $I$ either does not terminates within $t$ steps, or it terminates returning a sequence in which $x$ appears after $y$. Let $U$ be the set of all the possible instances and 
$U'$ be the set of the instances $I \in U$ such that $X_I = 0$ (i.e., $x$ and $y$ appear in the correct order in the output of $A$ on $I$).
Notice that there is a bijection between the instances in $I = \langle \pi, C \rangle$ and their corresponding instances $I_{\pi, C}$ (i.e., our transformation is injective). Moreover, since $I$ and $I_{\pi, C}$ are indistinguishable by $A_\lambda$, we have that
either (i) $A_\lambda$ does not terminate within $t$ steps on both instances, or (ii) $\sigma_I(x) < \sigma_I(y) \iff  \sigma_{I_{\pi, C}}(x) > \sigma_{I_{\pi, C}}(y)$. As a consequence, if $X_I = 0$ then  $X_{I_{\pi, C}} = 1$.
Now, either $\sum_{I \in U} X_I \Pr(I) \ge \frac{1}{2}$ or we must have $\sum_{I \in U'} \Pr(I) \ge \frac{1}{2}$, which implies:
\begin{align*}
	\sum_{I \in U} X_I \Pr(I)  &\ge \sum_{\langle \pi, C \rangle \in U} X_{I_{\pi, C}} \Pr(I_{\pi, C}) 
    \ge \sum_{\langle \pi, C \rangle \in U'} X_{I_{\pi, C}} \Pr(I_{\pi, C}) 
    = \sum_{\langle \pi, C \rangle \in U'} \Pr(I_{\pi, C}) \\
    & \ge \sum_{ I \in U'} \Pr( I ) \cdot \left( \frac{p}{1-p} \right)^{2(y-x)-1} 
    \ge \frac{1}{2} \left( \frac{p}{1-p} \right)^{2(y-x)-1}.
\end{align*}

We let $Y$ (resp. $Y_\lambda$) be an indicator random variable which is $1$ iff either the execution of $A$ (resp. $A_\lambda$) on a random instance does not terminate within $t$ steps, or it terminates returning a sequence in which $x$ and $y$ appear in the wrong relative order. 
By the above calculations we know that $\Pr(Y_\lambda = 1) = \sum_{I \in U} X_I \Pr(I) \ge \frac{1}{2} \left( \frac{p}{1-p} \right)^{2(y-x)-1} \; \forall \lambda \in \{0, 1\}^t$. We are now ready to bound the probability that $Y=1$, indeed:

\[
	\Pr(Y = 1) = \sum_{ \lambda \in  \{0, 1\}^t} \Pr(Y_\lambda=1) p_\lambda \ge  
	\frac{1}{2} \left( \frac{p}{1-p} \right)^{2(y-x)-1} \sum_{ \lambda \in \{0, 1\}^t } p_\lambda 
	= \frac{1}{2} \left( \frac{p}{1-p} \right)^{2(y-x)-1},
\]
\noindent where the last equality follows from the fact that $p_\lambda$ is a probability distribution over the elements $\lambda \in \{0, 1\}^t$, and hence  $\sum_{ \lambda \in \{0, 1\}^t } p_\lambda = 1$.
\qed

\subsection*{Proof of Theorem~\ref{thm:lb_tot_dislocation}}

	Let $A$ be any algorithm and let $\langle \pi,C \rangle$ be a random instance on an even number of elements $n$.
	We define $X_k$ for $0 \le k < n/2$ to be an indicator random variable that is $1$ iff $A(\langle \pi,C \rangle)$ 
	returns a permutation of the elements in $S$ in which elements $2k+1$ and $2k+2$ appear in the wrong order.
	
	By Lemma~\ref{lemma:swap_probability} we know that $\Pr(X_k=1) \ge \frac{1}{2}\frac{p}{1-p}$. We can hence obtain a lower bound on the expected total dislocation $\Delta$ achieved by $A$ as follows:
	\[
		\E[\Delta] \ge \sum_{k=0}^{n/2 - 1} \E[X_k] \ge \sum_{k=0}^{n/2 - 1} \frac{1}{2} \frac{p}{1-p} \ge \frac{n}{4} \frac{p}{1-p} = \Omega(n).
	\]
\qed

\end{document}